\def\year{2021}\relax

\documentclass[letterpaper]{article} 
\usepackage{aaai21}  
\usepackage{times}  
\usepackage{helvet} 
\usepackage{courier}  
\usepackage[hyphens]{url}  
\usepackage{graphicx} 
\usepackage{natbib}  
\usepackage{caption} 
\usepackage[switch]{lineno}

\PassOptionsToPackage{numbers}{natbib}
\usepackage{booktabs} 
\usepackage[ruled]{algorithm2e} 

\SetAlFnt{\small}
\SetAlCapFnt{\small}
\SetAlCapNameFnt{\small}
\SetAlCapHSkip{0pt}
\IncMargin{-\parindent}

\usepackage{graphicx}
\usepackage{amsmath,amssymb,amsthm}
\usepackage{algorithmic}
\usepackage{bm}
\usepackage{rotating}
\usepackage{multicol}
\usepackage{latexsym}
\usepackage{color}
\usepackage{mathtools}

\newcounter{newct}

\newcommand{\bv}{\begin{array}}

\newcommand{\E}{\mathbb{E}}

\newcommand{\appClaim}[2]{\noindent{\bf Claim~\ref{#1}.} {\em #2}}

\newcommand{\appThm}[3]{\noindent{\bf Theorem~\ref{#1} (#2).} {\em #3}}

\newtheorem{thm}{Theorem}
\newtheorem{dfn}{Definition}

\newtheorem{ex}{Example}

\newtheorem{prop}{Proposition}
\newtheorem{coro}{Corrollary}
\newenvironment{sketch}{\noindent{\em Proof sketch.}\rm }{\hfill $\Box$ }

\newtheorem{asmpt}{Assumption}
\newtheorem{claim}{Claim}

\newcommand{\ma}{\mathcal A}
\newcommand{\mm}{\mathcal M}

\newcommand{\ml}{\mathcal L}
\newcommand{\mV}{\mathcal V}

\newcommand{\ra}{\rightarrow}

\newcommand{\kt}{\text{KT}}
\newcommand{\avgkt}{\overline{\text{KT}}}

\newcommand{\mallows}{\text{Ma}}
\newcommand{\pl}{\text{Pl}}

\newcommand{\op}{\text{O}}

\newcommand{\expect}{{\mathbb E}}

\newcommand{\Omit}[1]{}

\newcommand{\md}{\mathcal D}

\newcommand{\gcyc}{{G_{3c}}}
\newcommand{\gco}{{G_{co}}}
\newcommand{\picyc}{{\pi_{3c}}}
\newcommand{\pico}{{\pi_{co}}}

\newcommand{\wmg}{\text{WMG}}
\newcommand{\ks}{{\sc Kemeny Score}}

\newcommand{\slaterranking}{{\sc Slater Ranking}}
\newcommand{\kr}{{\sc Kemeny Ranking}}
\newcommand{\brute}{\text{BF}}
\newcommand{\efas}{\text{\sc EFAS}}
\newcommand{\tfas}{\text{\sc TFAS}}

\newcommand{\spclass}{\text{\sc Smoothed-P}}
\newcommand{\snpclass}{\text{\sc Dist-NP}$_\text{para}$}

\newcommand{\umg}{\text{UMG}}

\newcommand{\alg}{\text{Alg}}
\newcommand{\rt}[1]{\text{Time}_{#1}}
\newcommand{\sr}[2]{\tilde{\text{RT}}_{#1}(#2)}

\title{The Smoothed Complexity of Computing Kemeny and Slater Rankings}
\author{Lirong Xia\textsuperscript{\rm 1} and  Weiqiang Zheng\textsuperscript{\rm 2}\\}
\affiliations {
\textsuperscript{\rm 1} RPI, \textsuperscript{\rm 2} Peking University \\
xialirong@gmail.com, weiqiang.zheng@pku.edu.cn
}

\begin{document}
\maketitle
\begin{abstract}
The computational complexity of {\em winner determination} under common voting rules is a classical and fundamental topic in the field of computational social choice. Previous work has established the {\sc NP}-hardness of winner determination under some commonly-studied voting rules, especially the Kemeny rule and the Slater rule. In a recent blue-sky paper, \citet{Baumeister2020:Towards} questioned the relevance of  the worst-case nature of {\sc NP}-hardness in social choice and proposed to conduct smoothed complexity analysis~\cite{Spielman2009:Smoothed} under~\citet{Blaser2015:Smoothed}'s framework.

In this paper, we develop the first smoothed complexity results for winner determination in voting. We illustrate the inappropriateness of~\citet{Blaser2015:Smoothed}'s smoothed complexity framework  in social choice contexts by proving a paradoxical result, which states that the exponential-time brute force search algorithm is smoothed poly-time according to their definition. We then prove the smoothed hardness of Kemeny and Slater using the classical smoothed complexity analysis, and prove a parameterized typical-case smoothed easiness result for Kemeny. 
Overall, our results show that smoothed complexity analysis in computational social choice is a challenging and fruitful topic.
\end{abstract}

\section{Introduction}

The computational complexity of {\em winner determination} under common voting rules is a classical and important topic in the field of computational social choice~\cite[Section 1.2.3]{Brandt2016:Handbook}. A low computational complexity of winner determination is desirable and is indeed the case for many commonly-studied and widely-applied voting rules. On the other hand, winner determination has been proved to be NP-hard for some classical voting rules such as the Kemeny rule and the Slater rule~\cite{Bartholdi89:Voting,Conitzer06:Slater}.

To address the worst-case nature of NP-hardness, average-case analysis has been conducted to provide a more realistic analysis of algorithms. However, average-case analysis is sensitive to the choice of the distribution over input instances, which may itself be unrealistic. To tackle this problem, \citet{Spielman2004:Smoothed} introduced {\em smoothed complexity analysis} to generalize and combine the worst-case analysis and the average-case analysis. The idea is that the  input $\vec x$ of an algorithm $\alg$ is often a noisy perception of the ground truth  input $\vec x^*$. Consequently, the worst-case is analyzed by assuming that an adversary chooses a ground truth $\vec x^*$ and then Nature adds a noise $\vec\epsilon$ (e.g.~a Gaussian noise) to it, such that the algorithm's input becomes $\vec x=\vec x^*+\vec \epsilon$. The smoothed runtime of $\alg$ is therefore defined as $\max_{\vec x^*}\expect_{\vec\epsilon}\  \rt{\alg}(\vec x^*+\vec\epsilon)$, in contrast to the worst-case runtime $\max_{\vec x^*}\rt{\alg}(\vec x^*)$ and the average-case runtime $\expect_{\vec x^*\sim\pi} \rt{\alg}(\vec x^*)$ under a given distribution $\pi$ over inputs.

Smoothed complexity analysis has been applied to a wide range of problems in mathematical programming, machine learning, numerical analysis, discrete math, combinatorial optimization, and equilibrium analysis  and price of anarchy, see the survey by~\citet{Spielman2009:Smoothed}.  In a recent blue-sky paper, \citet{Baumeister2020:Towards} proposed to conduct smoothed complexity analysis in computational social choice under~\citet{Blaser2015:Smoothed}'s framework and proposed a natural noise model that leverages the celebrated~\citet{Mallows57:Non-null} model. However, we are not aware of a technical result on the smoothed complexity of winner determination in voting. The following question remains open.

\begin{center}
{\em  What is the smoothed complexity of winner determination under commonly-studied voting rules?}
\end{center}

\noindent{\bf Our Model.}
We answer the question for the Kemeny rule and the Slater rule, for which winner determination means computing an optimal consensus ranking. We adopt the {\em smoothed social choice framework}~\cite{Xia2020:The-Smoothed}, which covers a wide range of models, including the Mallows-based model proposed by~\citet{Baumeister2020:Towards}. In the  framework, for any  number of alternatives and any number of agents, denoted by $m\ge 3$ and $n\ge 1$ respectively, the adversary chooses a distribution $\pi_j$  for each agent $j$ from a set of distributions $\Pi_m$ over all rankings. Let $\vec \pi=(\pi_1,\ldots,\pi_n)\in \Pi_m^n$. Then, given an algorithm $\alg$, the adversary aims at choosing $\vec\pi$ to maximize the expected runtime of $\alg$ on the profile where each ranking is generated independently from $\vec \pi$, defined as follows.
\begin{equation}
\label{eq:classicalsmoothed}
\sr{\Pi_m}{\alg, m, n}=\sup_{\vec \pi\in \Pi_m^n}\expect_{P\sim \vec\pi}\  \rt{\alg}(P)
\end{equation}
Because the size of input of the algorithm is the size of $P$, i.e.~$\Theta(nm\log m)$,  we desire $\sr{\Pi_m}{\alg,m,n}$ to be polynomial in $m$ and $n$. Following the convention in statistics and the notation in~\cite{Xia2020:The-Smoothed}, for any $m\ge 3$, we use a {\em single-agent preference model} $\mm_m=(\Theta_m,\ml(\ma_m),\Pi_m)$ to model the adversary's capability, where $\Theta_m$ is the parameter space and $\ml(\ma_m)$ is the set of all rankings over $m$ alternatives.  

Note that when $m$ is a constant, many commonly-studied voting rules, including Kemeny and Slater, are easy to compute. Therefore, to meaningfully analyze the smoothed complexity, in this paper we are given an infinite series of models $\vec\mm=\{\mm_m=(\Theta_m,\ml(\ma_m),\Pi_m):m\ge 3\}$, following the convention in average-case complexity theory~\cite{Bogdanov2006:Average-Case} and the smoothed complexity theory proposed by~\citet{Blaser2015:Smoothed}. 


\vspace{2mm}\noindent{\bf Our Contributions.} We first prove a paradoxical result (Proposition~\ref{prop:brute}) under \citet{Blaser2015:Smoothed}'s framework to show its inappropriateness. The proposition states that the exponential-time brute force search algorithm for Kemeny and Slater is smoothed poly-time w.r.t.~a large class of models including the model proposed by~\citet{Baumeister2020:Towards}.  

Consequently, we  conduct the classical smoothed complexity analysis~\cite{Spielman2009:Smoothed} and prove two smoothed hardness results for Kemeny (Theorem~\ref{thm:kemenyhard}) and Slater  (Theorem~\ref{thm:slaterhard}), respectively. Both theorems state that for a large class of models, if a smoothed poly-time algorithm exists, then RP$=$NP. 

Finally, we consider parameterized typical-case smoothed complexity and prove a mildly positive result in Theorem~\ref{thm:kemenyeasy}, which implies that for a large class of Mallows-based models, if the average Kendall Tau's distance in the central rankings and the average dispersion parameters are not too large, then the dynamic programming algorithm for Kemeny proposed by~\citet{Betzler09:Fixed} runs in poly-time with high probability.

\vspace{2mm}\noindent{\bf Related Work and Discussions.}  We are not aware of a previous technical result about the smoothed complexity of social choice problems. As discussed above, \citet{Baumeister2020:Towards} proposed to conduct smoothed complexity analysis in computational social choice and proposed a natural Mallows-based model for such analysis. \citet{Xia2020:The-Smoothed} showed that some paradoxes and impossibility theorems in social choice vanishes in the smoothed sense. 

Our paper aims at making the first technical attempt of smoothed complexity analysis in computational social choice, and overall our results show that  the topic is highly challenging and fruitful. The seemingly paradoxical smoothed efficiency of brute force search under~\citet{Blaser2015:Smoothed}'s  framework  (Proposition~\ref{prop:brute}) does not mean that their framework is ``wrong'', but instead, should be interpreted as a call  for future theories of smoothed complexity analysis in social choice contexts. The smoothed hardness of Kemeny (Theorem~\ref{thm:kemenyhard}) and Slater (Theorem~\ref{thm:slaterhard})
are negative news and the parameterized typical-case smoothed efficiency (Theorem~\ref{thm:kemenyeasy}) is positive news. Proof techniques developed for these theorems may be useful in future work.

There is a large body of literature on the computational complexity of Kemeny. The corresponding optimization problem, \kr{}, was proved to be {\sc NP}-hard~\cite{Bartholdi89:Voting} and P$_\|^\text{\sc NP}$-complete~\cite{Hemaspaandra05:Complexity}. 
Approximation algorithms~\cite{Ailon2008:Aggregating,Zuylen2007:Deterministic}, PTAS~\cite{Kenyon07:How}, and fixed-parameter efficient algorithms~\cite{Betzler09:Fixed,Karpinski2010:Faster,Cornaz2013:Kemeny} for \kr{} have been developed. Practical algorithms have been proposed~\cite{Davenport04:Computational,Conitzer06:Kemeny} and \citet{Ali12:Experiments} compared the performance of 104 algorithms.  \citet{Conitzer06:Slater} proved that the decision variant of Slater is {\sc NP}-hard and proposed an efficient heuristic algorithm for computing \slaterranking{}.

There is a large body of literature on smoothed complexity of algorithms~\cite{Spielman2009:Smoothed}. \citet{Blaser2015:Smoothed} established a complexity theory for smoothed complexity analysis by defining the counterparts to P and NP, called \spclass{} and \snpclass{} respectively, together with a smoothed reduction and complete problems. Their definitions are closely related to the average-case complexity theory established by~\citet{Levin1986:Average}. However, as we will show in Section~\ref{sec:brute}, it may not be suitable for computational social choice. Our Theorem~\ref{thm:kemenyhard} and~\ref{thm:slaterhard} illustrate the hardness of \kr{} and \slaterranking{} using the same pattern in~\cite{Huang2007:On-the-Approximation}, which states that the existence of a smoothed poly-time algorithm would lead to a surprise in complexity theory.



\section{Preliminaries}

\noindent{\bf Basics of Voting.} For any $m\ge 3$, let $\ma_m=\{a_1,\ldots,a_m\}$ denote the set of $m$ alternatives.  A {\em (preference) profile} $P\in \ml(\ma_m)^n$ is a collection of $n$ rankings (linear orders). Let $\wmg(P)$ denote the {\em weighted majority graph} of $P$, which is a directed weighted graph whose vertices are $\ma_m$ and for each pair of alternatives $a, b$, the weight on edge $a\ra b$, denoted by  $w_P(a,b)$, is the winning margin of $a$ over $b$ in their pairwise competition. That is, $w_P(a,b) = -w_P(b,a)=\#\{R\in P: a\succ_R b\}-\#\{R\in P: b\succ_R a\}$.  Let $\umg(P)$ denote the {\em unweighted majority graph} of $P$, which is the unweighted directed graph obtained from $\wmg(P)$ by removing edges whose weights are $\le 0$. 

The {\em Kendall's Tau} distance between two linear orders $R,W\in\ml(\ma)$, denoted by $\kt(R,W)$,  is the number of pairwise disagreements between $R$ and $W$. Given a profile $P$ and a linear order $R$, the {\em Kemeny score} of $R$ in $P$ is $\sum_{W\in P}\kt(R,W)$ and the {\em Slater score}  of $R$ in $P$ is $\kt(R,\umg(P))$, where $\kt$ is extended to measure the distance between a linear order and an unweighted graph in the natural way. The Kemeny rule (respectively, the Slater rule) aims at selecting the linear order with the minimum Kemeny score (respectively, Slater score) in $P$. The corresponding winner determination problems are defined as follows.

\begin{dfn}[\bf \kr{} and \slaterranking{}] Given $m\ge 3$, $n\in\mathbb N$, and $P\in \ml(\ma_m)^n$, in \kr{} (respectively, \slaterranking{}), we are asked to compute a ranking with minimum Kemeny score (respectively, Slater score).
\end{dfn}

\begin{dfn}[\bf Single-agent preference model~\cite{Xia2020:The-Smoothed}]A {\em single-agent preference model} for $m$ alternatives is denoted by $\mm_m=(\Theta_m,\ml(\ma_m),\Pi_m)$, where $\Pi_m$ is the set of distributions over $\ml(\ma_m)$ indexed by the parameter space $\Theta_m$.
 $\mm_m$ is {\em neutral} if for any $\theta\in\Theta_m$ and any permutation $\sigma$ over $\ma_m$, there exists $\theta'\in\Theta_m$ such that for all $V\in\ml(\ma_m)$, we have $\pi_\theta(V) = \pi_{\theta'}(\sigma(V))$.  $\mm_m$ is {\em P-samplable} if there exists a poly-time sampling algorithm for each distribution in $\Pi_m$.
\end{dfn}
Technically $\mm_m$ is completely determined by $\Pi_m$. Following the convention in statistics, we sill keep the parameter space $\Theta_m$ and sample space $\ml(\ma_m)$ in the definition.

\begin{ex}
\label{ex:singlemallows}
In a  {\em single-agent Mallows model}  $\mm_{\mallows,m}$, $\Theta_m = \ml(\ma_m)\times (0,1]$, where in each $(R,\varphi)\in \Theta_m$, $R$ is called the {\em central ranking} and $\varphi$ is called the {\em dispersion parameter}. For any $W\in \ml(\ma_m)$, we have $\pi_{(R,\varphi)} = \varphi^{\kt(R,W)}/Z_{\varphi}$, where $Z_\varphi = \frac{\prod_{i=2}^m(1-\varphi^i)}{(1-\varphi)^{m-1}}$ 
is the normalization constant. For any $0< \underline\varphi\le \overline\varphi\le 1$, we let $\mm_{\mallows,m}^{[\underline\varphi, \overline\varphi]}$ denote the sub-model whose parameter space is  $\ml(\ma_m)\times [\underline\varphi, \overline\varphi]$. $\mm_{\mallows,m}^{[\underline\varphi, \overline\varphi]}$ is  neutral and P-samplable~\cite{Doignon2004:Repeated}.

See Appendix~\ref{app:pl} for another class of neutral and P-samplable models based on the Plackett-Luce model. 
\end{ex}

When there are $n\ge 2$ agents, the adversary chooses $\vec\pi=(\pi_1,\ldots,\pi_n)\in \Pi_m^n$, and then agent $j$'s ranking will be independently (but not necessarily identically) generated from $\pi_j$.  

\begin{ex}\label{ex:models}
Suppose $m=3$ and $n=2$, and the model is $\mm_{\mallows,3}^{[0.3,1]}$. Then, the adversary can set the first (respectively, second) agent's distribution to be the Mallows distribution given ground truth $(a_1\succ a_2\succ a_3,0.4)$ (respectively, $( a_3\succ a_2\succ a_1,0.8)$). Then, the probability of generating profile $( a_2\succ a_1\succ  a_3,  a_1\succ a_3\succ  a_2)$ is $\Pr( a_2\succ a_1\succ  a_3|( a_1\succ a_2\succ a_3,0.4))\times \Pr( a_1\succ a_3\succ  a_2|( a_3\succ a_2\succ a_1,0.8)) = \frac{0.4}{Z_{0.4}}\times  \frac{0.8^2}{Z_{0.8}}$.

As another example, the Mallows-based model proposed by~\citet{Baumeister2020:Towards} corresponds to the single-agent Mallows model with fixed  $\varphi$, i.e.~$\mm_{\mallows,m}^{[\varphi,\varphi]}$.
\end{ex}


Because \ks{} is in P when $m$ is bounded above by a constant, the smoothed complexity analysis ought to be done for variable $m$. Therefore, following~\cite{Blaser2015:Smoothed}, we are given a series of single-agent preference models $\vec\mm=\{\mm_m=(\Theta_m,\ml(\ma_m),\Pi_m):m\ge 3\}$.  For example, for any $0<\underline\varphi \le \overline\varphi\le 1$, we let 
$\vec\mm_\mallows^{[\underline\varphi,\overline\varphi]}=\{\mm_{\mallows,m}^{[\underline\varphi,\overline\varphi]}:m\ge 3\}$ denote a series of Mallows-base models.

In most part of this paper (except Section~\ref{sec:brute}), we focus on the classical smoothed poly-time algorithms proposed by~\citet{Spielman2009:Smoothed} w.r.t.~$\vec\mm$. 

\begin{dfn}[\bf Smoothed poly-time]\label{dfn:spt} Given a series of single-agent preference models $\vec\mm =\{\mm_m=(\Theta_m,\ml(\ma_m),\Pi_m):m\ge 3\}$,  an algorithm $\alg$ is smoothed poly-time, if for any $m\ge 3$ and $n\ge 1$, $\sr{\Pi_m}{\alg, m, n}$ as defined in~(\ref{eq:classicalsmoothed}) is polynomial in $m$ and $n$. \end{dfn}

\noindent{\bf \citet{Blaser2015:Smoothed}'s framework.} In the framework, the smoothed runtime of an algorithm is analyzed w.r.t.~a set of distributions $\md=\{D_{\ell,x,\phi}\}$  over the input. For each distribution $D_{\ell,x,\phi}$, $x$ and $\phi$ are the parameters of the distribution, which can be viewed as the counterparts to the ground truth input and inverse variance in the Gaussian noise studied by~\citet{Spielman2009:Smoothed}. $\ell$ is the size of $x$, which is redundant but is included for clarity. Let $S_{\ell,x}$ denote the set of inputs, each of which receives a positive probability under $D_{\ell,x,\phi}$ for some $\phi$. Let $N_{\ell,x}=|S_{\ell,x}|$. \citet{Blaser2015:Smoothed}  defined the following notion of smoothed poly-time algorithms by further requiring that (1) the size of each element in $S_{\ell,x}$ is bounded above by a polynomial of $\ell$, (2) for each $y$, $D_{\ell,x,\phi}(y)<\phi$, and (3) $\phi$ is discretized and can be described by $\text{poly}(\ell)$ bits. Any $\md$ satisfying these conditions is called a {\em perturbation model}.

\begin{dfn}[\bf BM-Smoothed poly-time~\cite{Blaser2015:Smoothed}]\label{dfn:smoothedpoly}  Given a perturbation model $\md=\{D_{\ell,x,\phi}\}$, an algorithm $\alg$ is {\em BM-smoothed poly-time} if there exists $\epsilon>0$ such that  for all $D_{\ell,x,\phi}\in \md$, 
\begin{equation}\label{eq:smoothedpoly}
\expect_{y\sim D_{\ell,x,\phi}}(\rt{\alg}(y)^\epsilon) = O(n\cdot \phi \cdot N_{\ell,x})
\end{equation}
\end{dfn}
The exponent $\epsilon$ is defined following the average-case complexity theory, which implies that $\Pr_{y\sim D_{\ell,x,\phi}}(\rt{\alg}(y)\ge t) \le \frac{\text{poly}(\ell)}{t^\epsilon}N_{\ell,x}\phi$~\cite[Thm.~2.3]{Blaser2015:Smoothed}. Note that being BM-smoothed poly-time does not imply that the algorithm has polynomial smoothed runtime under~\citet{Spielman2009:Smoothed}'s framework. $N_{\ell,x}$ in the right hand side is introduced for a technical reason, which allows the expected runtime to be exponentially large if $N_{\ell,x}\phi$ is exponentially large. See~\cite{Blaser2015:Smoothed}  for more discussions.


\section{Brute Force Search is BM-Smoothed Poly-Time for Kemeny and Slater}
\label{sec:brute}
Let $\brute$ denote the brute force search algorithm that first computes the Kemeny scores (respectively, Slater scores) for all $m!$ rankings, and then chooses the one with the minimum score. The following proposition states that $\brute$ is BM-smoothed poly-time for a large class of models.

\begin{prop}\label{prop:brute} For any fixed $0<\underline\varphi \le \overline\varphi<1$,  when $m\ge 2^{(3-\overline\varphi)/(1-\overline\varphi)}$, \brute{} is BM-smoothed poly-time for \kr{} and \slaterranking{} w.r.t.~$\vec\mm_{\mallows}^{[\underline\varphi, \overline\varphi]}$, where the dispersion parameter is discretized.
\end{prop}
\begin{proof}
We first translate $\vec\mm_{\mallows}^{[\underline\varphi, \overline\varphi]}$ to the $D_{\ell,x,\phi}$ notation. For any $m$, $n$, and any $\vec\pi\in \Pi_{m}^n$, let $x=\vec\pi$ represent the central rankings and dispersion parameters for the $n$ agents. Therefore, $\ell = \Theta(nm\log m)$, $N_{\ell,x} = (m!)^n$, and $\phi \ge (\frac{1}{Z_{\underline\varphi}})^n\ge (1-\overline\varphi)^{n(m-1)}$. Let $\epsilon = \frac 12$. For any $n$-profile $P'\in\ml(\ma_m)^n$, $\rt{\brute}(P') = O(m!nm^2)$ for both \kr{} and \slaterranking{}. Therefore, $\expect_{P'\sim D_{\ell,x,\phi}}(\rt{\brute}(P')^\epsilon) = O((m!nm^2)^\epsilon)$. To prove the proposition, it suffices to prove that $nm\log m (m!(1-\overline\varphi)^{m-1})^n>(m!nm^2)^{1/2}$ for any $m>2^{(3-\overline\varphi)/(1-\overline\varphi)}$ and $n\ge 1$. This is proved by the following series of inequalities.
\begin{align*}
& m>2^{(3-\overline\varphi)/(1-\overline\varphi)}\Leftrightarrow m\log m > m(\frac{2}{1-\overline\varphi}+1)\\
\Rightarrow & (m+\frac 12) \log m -m >(m-1)\frac{2}{1-\overline\varphi}\\
\Rightarrow &\log m! > (m-1)\frac{2}{1-\overline\varphi} \ \ \ \text{(Stirling's formula)}\\
\Rightarrow & (m!)^{1-\frac{1}{2n}}(1-\overline\varphi)^{m-1}>1\\
\Rightarrow & nm\log m (m!(1-\overline\varphi)^{m-1})^n > (m!nm^2)^{1/2}
\end{align*}
\end{proof}

The paradox illustrated in Proposition~\ref{prop:brute} comes from two sources. The first source is the  $N_{\ell,x}\phi$ term in (\ref{eq:smoothedpoly}). This means that the expected runtime$^\epsilon$ is allowed to be exponentially large when the perturbation is small. Therefore, the problem becomes interesting only when the perturbation is large. This is in sharp contrast to the classical smoothed complexity analysis, where~\citet{Spielman2009:Smoothed} are ``most interested in ... slight perturbation". In Proposition~\ref{prop:brute}, $\vec\mm_{\mallows}^{[\underline\varphi, \overline\varphi]}$ can be viewed as slight perturbations when $n$ is large. Therefore, our setting is closer to \citet{Spielman2009:Smoothed}'s setting than to~\citet{Blaser2015:Smoothed}'s setting.

The second source is the $\epsilon$ exponent in (\ref{eq:smoothedpoly}), which aims at capturing typical runtime as in the average-case complexity theory~\cite{Bogdanov2006:Average-Case} instead of the averaged runtime  as in the classical smoothed  complexity  analysis~\cite{Spielman2009:Smoothed}.

\section{Smoothed Hardness of Kemeny and Slater}

In this section, we follow the classical smoothed runtime analysis (Definition~\ref{dfn:spt})  to analyze  \kr{} and \slaterranking{} w.r.t.~models that satisfy Assumption~\ref{asmpt:main} below. We first recall the {\em orthogonal decomposition} of weighted directed graphs~\cite{Young74:Axiomatization,Zwicker2018:Cycles}. 

A WMG $G_\text{cyc}$ is called a {\em cycle}, if the absolute weight of any edge is $0$ or $1$, and the  edges with positive weights forms a cycle. Let $a\in\ma_m$, a WMG $G_a$ is called a {\em co-cycle} centered at $a$, if for any $b\in\ma_m$ such that $b\ne a$, $w_{G_a}(a,b)=-w_{G_a}(b,a) =1$ and all other edges have weight $0$. 

Note that any WMG $G$ can be viewed as a vector in ${\mathbb R}^{\frac{m(m-1)}{2}}$ whose components are indexed by $(i_1,i_2)$, where $1\le i_1<i_2\le m$, with value $w_{G}(a_{i_1},a_{i_2})$. Given any pair of WMGs $G_1$ and $G_2$, we define their dot product as
$$G_1\cdot G_2 = \sum_{1\le i_1<i_2\le m}w_{G_1}(a_{i_1},a_{i_2})\times w_{G_2}(a_{i_1},a_{i_2})$$ 

Let $\mV_{cyc}\subseteq {\mathbb R}^{\frac{m(m+1)}{2}}$ (respectively, $\mV_{co}\subseteq{\mathbb R}^{\frac{m(m+1)}{2}}$) denote the linear span of cycles (respectively, co-cycles). It has been proved that $\mV_{cyc}$ and $\mV_{co}$ are orthogonal, $\dim(\mV_{cyc}) = {m-1\choose 2}$ (all 3-cycles containing $a_1$ in the increasing direction of subscripts constitute a  non-orthogonal basis) and $\dim(\mV_{co}) = m-1$ (co-cycles centered at any fixed set of $m-1$ alternatives constitute a  non-orthogonal basis). An {\em orthogonal decomposition} of a WMG $G$ is a decomposition of $G$ into its projections to $\mV_{cyc}$ and $\mV_{co}$, respectively.  



\begin{ex} Let $\theta=(a_1\succ a_2\succ a_3,\varphi)$ denote a parameter in $\mm_{\mallows,3}^{[\underline\varphi,\overline\varphi]}$.  $\wmg(\theta)$, which is the WMG of the fractional profile represented by the distribution corresponding to $\theta$, and its orthogonal decomposition are shown in Figure~\ref{fig:wmg}. Note that the weight on the co-cycle centered at $a_3$ is negative.

\begin{figure}[htp]
\centering
\includegraphics[width=.5\textwidth]{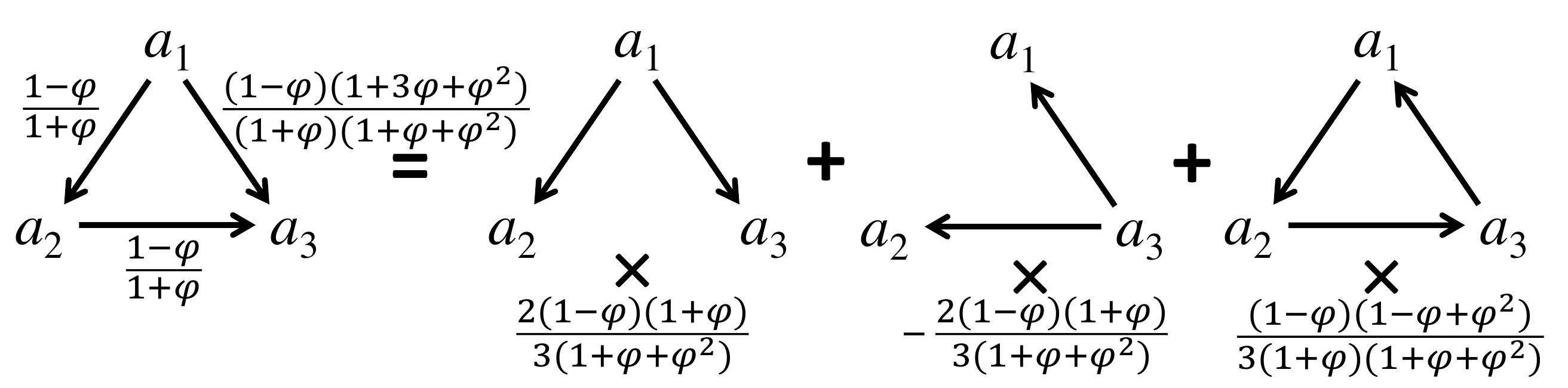} 
\caption{\small The WMG of $(a_1\succ a_2\succ a_3,\varphi)$ in $\mm_{\mallows,3}$ and its orthogonal decomposition.\label{fig:wmg}}
\end{figure}
\end{ex}


\begin{asmpt}\label{asmpt:main} $\vec\mm=\{\mm_m=(\Theta_m,\ml(\ma_m),\Pi_m):m\ge 3\}$ satisfies the following conditions:

(i) For any $m\ge 3$, $\mm_m$ is {\em P-samplable}, i.e.~it admits a poly-time sampling algorithm.

(ii) For any $m\ge 3$, $\mm_m$ is neutral, i.e.~for each distribution $\pi\in\Pi_m$ and any distribution $\sigma$ over the alternatives, we have $\sigma(\pi)\in\Pi_m$.

(iii) There exist constants $k\ge 0$ and $A>0$ such that for any $m\ge 3$, there exist $\pi_{3c}\in \Pi_m$ such that $\wmg(\pi_{3c})$ has a 3-cycle component $\gcyc$ with $\wmg(\pi_{3c})\cdot \gcyc>\frac{A}{m^{k}}$.
\end{asmpt}
The first condition is the ``most natural restriction'' on general distributions~\cite[p.~17]{Bogdanov2006:Average-Case}, which is less restrictive than the commonly-studied {\em P-computable} distributions~\cite[p.~18]{Bogdanov2006:Average-Case}. The second and third conditions imply that $\vec\mm$ is rich enough. 
As we will see Example~\ref{ex:assumption-hold-mallows}, a large class of models, including the model proposed by~\citet{Baumeister2020:Towards}, satisfy Assumption~\ref{asmpt:main}.

\begin{thm}[\bf Smoothed Hardness of \kr{}] \label{thm:kemenyhard} For any series of single-agent preference models $\vec\mm$ that satisfies Assumption~\ref{asmpt:main}, if there exists a smoothed poly-time algorithm for \kr{}  w.r.t.~$\vec\mm$, then {\sc NP}$=${\sc RP}.
\end{thm}
\begin{proof} The theorem is proved by contradiction. Suppose for the sake of contradiction that a smoothed poly-time algorithm $\alg$ for \kr{} exists, we will prove that there exists an efficient randomized algorithm for the NP-hard problem {\sc Eulerian Feedback Arc Set (EFAS)}~\cite{Perrot2015:Feedback}. An instance of \efas{} is denoted by $(G,t)$, where $t\in\mathbb N$ and $G$ is a directed unweighted Eulerian graph, which means that there exists a Eulerian cycle that passes each edge exactly once. We are asked to decide whether $G$ can be made acyclic by removing no more than $t$ edges.

Given a single-agent preference model, a (fractional) {\em parameter profile} $P^\Theta\in \Theta_m^n$ is a collection of $n>0$ parameters, where $n$ may not be an integer. Note that $P^\Theta$ naturally leads to a fractional preference profile, where the weight on each ranking represents its total weighted ``probability'' under all parameters in $P^\Theta$. Therefore, $\wmg$ and $\umg$ can be naturally extended to parameter profiles. 

The high-level idea of the proof is the following. For any \efas{} instance $(G,t)$,  we will construct a fractional parameter profile $P^\Theta_{G}$ whose WMG is the same as the WMG equivalent of $G$, where the weight on each edge in $G$ is $1$ or $-1$, depending on its direction. Then, we sample a profile $P'$ from $P^\Theta_G$ and try to run $\alg$ to compute the Kemeny ranking $R^*$. If $\alg$ successfully returns $R^*$ in less than three times of its expected runtime (which is polynomial), then we proceed to check whether $R^*$ leads to a YES answer to $(G,t)$. We give a NO answer if $R^*$ cannot be obtained from $G$ by removing no more than $t$ edges, or $\alg$ fails to terminate in time. Clearly this polynomial-time procedure always returns NO if $(G,t)$ is a  NO instances.  We then prove that a YES instance will receive a YES  answer with probability at least $1/2$, which means that \efas{} is in RP and therefore proves the theorem.

Formally, the proof proceeds in three steps. 
In Step 1, we use permutations of $\picyc$ guaranteed by Assumption~\ref{asmpt:main} to construct a fractional parameter profile $P^\Theta_\gcyc$ whose WMG is a 3-cycle. In Step 2, we use $P^\Theta_\gcyc$  to construct $P^\Theta_G$. In Step 3, we show that $\alg$ can be leveraged to Algorithm~\ref{alg:rpalg} to prove that  \efas{} is in RP as discussed above.

\noindent{\bf Step 1.~Construct a parameter profile $P^\Theta_\gcyc$ whose WMG is $\bm{\gcyc}$.} W.l.o.g.~let $\gcyc= a_1\ra a_2\ra a_3\ra a_1$. 
Let $\sigma_1$ denote an arbitrary cyclic permutation among $\{a_1,a_2,a_3\}$ and let $\sigma_2$ denote an arbitrary cyclic permutation among other alternatives. We define the following set of $6(m-3)$ permutations. 
$$\op_\gcyc = \{ \sigma_1^i\circ \sigma_2^t, \sigma_1^{i}\circ\sigma_2^{-t}:1\le i\le 3, 1\le t\le m-3\}$$
We note that $\op_\gcyc$ can be naturally applied to   linear orders,  fractional profiles, distributions over $\ml(\ma_m)$, parameters in $\Theta_m$, and weighted majority graphs. Let 
$\theta_{3c}$ denote the parameter corresponding to  $\pi_{3c}$  and let $Q^\Theta_\gcyc = \op_\gcyc(\theta_{3c})$ denote the parameter profile obtained from $\theta_{3c}$ by applying permutations in $O_\gcyc$. Some properties of  $Q^\Theta_\gcyc$ are described in the following claim, whose proof follows after definition. All missing proofs can be found in Appendix~\ref{app:proofs}.

\begin{claim}\label{claim:qtheta}
$|Q^\Theta_\gcyc| = O(m)$, $Q^\Theta_\gcyc$ consists of $O(m)$ types of parameters. $\wmg(Q^\Theta_\gcyc)$ consists of the following two types of edges.
(1) There are three edges $a_1\ra a_2$, $a_2\ra a_3$, $a_3\ra a_1$, each has weight $2(m-3)\alpha$, where $\alpha = \wmg(\picyc)\cdot \gcyc$.
(2) There are edges from $\{a_1,a_2,a_3\}$ to $\{a_4,\ldots,a_m\}$ whose weights are $\beta = 2\sum_{(d_1,d_2)\in \{a_1,a_2,a_3\}\times \{a_4,\ldots,a_m\}} w_{\pi_{3c}}(d_1,d_2)$.
\end{claim}
$|Q^\Theta_\gcyc|$ and the types of parameters used in $Q^\Theta_\gcyc$ will be crucial later in proving that Algorithm~\ref{alg:rpalg} runs in polynomial time. If $\beta =0$ then we let $P^\Theta_\gcyc = Q^\Theta_\gcyc$. 
If $\beta>0$, then we will provide a gadget soon to ``cancel'' edges between $\{a_1,a_2,a_3\}$ and $\ma_m\backslash\{a_1,a_2,a_3\}$. 
For any alternative $a$, let $\sigma_a$ denote the permutation such that $\sigma_a(a) = a$
and the remaining alternatives are permuted in the cyclic way in the increasing order of their subscripts. That is, $\sigma_a = a_{i_1}\ra a_{i_2}\ra \cdots\ra a_{i_{m-1}}\ra a_{i_1}$, where $i_1<\cdots<i_{m-1}$. We let $\eta$ denote the permutation that switches $a_{i_s}$ and $a_{i_{m-s}}$ for all $1\le s\le m-1$. For example, when $m=5$, $\sigma_{a_1}$ is the cyclic permutation $a_2\ra a_3\ra a_4\ra a_5\ra a_2$ and $\eta_{a_1}$ switches two pairs of alternatives: $(a_2, a_5)$ and $(a_3,a_4)$. 

\begin{dfn}\label{dfn:op}
For any alternative $a$, let $\op_a$ denote the following set of $2(m-1)$ permutations over $\ma_m$.
$$\op_a  = \{ \sigma_a^i , \sigma_a^{i}\circ \eta_a : 1\le i\le m-1\}$$
\end{dfn}

Like $\op_\gcyc$, $\op_a$ can be naturally applied to linear orders,  fractional profiles, distributions over $\ml(\ma_m)$, parameters in $\Theta_m$, and weighted majority graphs. Let $P_a^\Theta = \op_a(\theta_{3c})$. We have the following claim.

\begin{claim}
\label{claim:patheta} $|P_a^\Theta|=2(m-1)$ and 
$P_a^\Theta$ consists of $2(m-1)$ types of parameters. $\wmg(P^\Theta_a)$ is a co-cycle whose center is $a$ and the absolute weight on any non-zero edge is $2\sum_{(d_1,d_2)\in \{a\}\times {\ma \backslash \{a\}}}w_{\pi_{3c}}(d_1,d_2)$.
\end{claim}

Let $P_1^\Theta = \op_{a_1}(Q_{\gcyc}^\Theta)$.

\begin{claim}\label{claim:cocycle} $|P_1^\Theta|=O(m^2)$ and $P_1^\Theta$ consists of $O(m^2)$ types of parameters. $\wmg(P^\Theta_1)$ is a co-cycle whose center is $a_1$ and the absolute weight on any non-zero edge is $2(m-3)\beta$.
\end{claim}

For any pair of alternative $a,d$, let $\sigma_{a\leftrightarrow d}$ denote the permutation that exchanges $a$ and $d$. It follows that $\sigma_{a\leftrightarrow d}(P^\Theta_1)$ is a parameter profile whose WMG is a co-cycle centered at $d$. 
We are now ready to define the fractional parameter profile $P^\Theta_\gcyc$ whose \wmg{} resembles $\gcyc$. 
\begin{dfn} 
Let $P^\Theta_\gcyc$ denote the fractional parameter profile that consists of  (1) $\frac{1}{2(m-3)\alpha}$ copies of $Q^\Theta_\gcyc$, and (2)
for each $d\in\{a_4,\ldots,a_m\}$,  $\frac{1}{4(m-3)^2\alpha}$ copies  of $\sigma_{a_1\leftrightarrow d}(P^\Theta_1)$.
\end{dfn}
We have the following claim about $P^\Theta_\gcyc$.

\begin{claim}\label{claim:3cycle}
$|P^\Theta_\gcyc|  = O(m^{k})$ and  
$P^\Theta_\gcyc$ consists of $O(m^3)$ different types of parameters. $\wmg(P^\Theta_\gcyc)=\gcyc$. 
\end{claim}

\noindent{\bf Step 2.~Construct a parameter profile $\bm{P_G^\Theta}$ for \efas{}.}  
Because $\mm_m$ is neutral, for any 3-cycle $G' = a_{i_1}\ra a_{i_2}\ra a_{i_3}\ra a_{i_1}$ we can apply a permutation $\sigma_{G'}$ that maps $a_s$ to $a_{i_s}$ ($s=1,2,3$) on $P_1^\Theta$, which means that $\wmg(\sigma(P_1^\Theta))$ resembles $G'$. 

It is not hard to verify that any cycle of length $T$ can be obtained from the union of $T-2$ $3$-cycles, which can be computed in $O(m^2)$ time. Therefore, given the \efas{} instance $(G,t)$, we first compute $G=\bigcup_{s=1}^S G_s'$, where each $G_s'$ is a $3$-cycle, and $S\le {m\choose 2}$.  Then, we let 
$P_G^\Theta = \bigcup_{s=1}^S \sigma_{G_s'}(P_1^\Theta)$. It follows from Claim~\ref{claim:3cycle} that $|P_G^\Theta| = O(m^{k+2})$, $P_G^\Theta$ consists of $O(m^5)$ types of parameters, and $\wmg(P_G^\Theta) = G$.

\noindent{\bf Step 3.~Use $\alg$ to solve \efas.} Let $K = 11+2k$, which means that $K>9$. We first define an integral parameter profile $P_G^{\Theta*}$ of $n=\Theta(m^{K})$ parameters that is approximately $\frac{m^{K}}{|P_G^{\Theta}|}$ copies of $P_G^\Theta$ up to $O(m^5)$ in  $L_\infty$ error. Formally, let 
\begin{equation}\label{eq:pg}
P_G^{\Theta*} = \lfloor P_G^\Theta \cdot \dfrac{m^{K}}{|P_G^{\Theta}|}\rfloor
\end{equation}

Let $n=|P_G^{\Theta*}| $. Because the number of different types of parameters in $P_G^{\Theta*}$ is $O(m^5)$, we have $n = m^{K}-O(m^5)$, $\|\wmg(P_G^{\Theta*}) - \wmg(P_G^{\Theta} \cdot \frac{m^{K}}{|P_G^{\Theta}|}) \|_\infty  = O(m^5)$, and $\|\wmg(P_G^{\Theta*}) - G \cdot \frac{m^{K}}{|P_G^{\Theta}|}) \|_\infty  = O(m^5)$. 

We now prove that $\alg$ can be leveraged to provide an RP algorithm (Algorithm~\ref{alg:rpalg}) for \efas{}.

\begin{algorithm}
\SetAlgoLined
{\bf Input}: An \efas{} instance $(G,t)$,  $\alg$ for \kr{} whose smoothed runtime is $T$.

Compute a parameter profile $P_G^{\Theta*}$ according to (\ref{eq:pg}).

Sample a profile $P'$ from $\vec \mm_m$ given $P_G^{\Theta*}$ and run $\alg$ on $P'$.

\eIf{$\alg$ returns $R^*$ within $3T$ time and $R^*$ is a solution to $(G,t)$}{
   {\bf return} YES}{
   {\bf return} NO
  }\caption{Algorithm for \efas{}.}
\label{alg:rpalg}
\end{algorithm}
Notice that sampling $P'$ from  $P_G^{\Theta*}$ takes polynomial time because $\vec\mm$ is P-samplable (Assumption~\ref{asmpt:main} (i)). It follows that Algorithm~\ref{alg:rpalg} is a polynomial-time randomized algorithm. Clearly if $(G,t)$ is a NO instance then Algorithm~\ref{alg:rpalg} returns NO. Therefore, to prove that Algorithm~\ref{alg:rpalg} is an RP algorithm it suffices to prove that if $(G,t)$ is a YES instance then Algorithm~\ref{alg:rpalg}  returns YES  with $>\frac 12$ probability. 

Let $G_n = G \cdot \frac{m^{K}}{|P_G^{\Theta}|}$.   We first prove in the following claim that with exponentially small probability $\wmg(P')$ is different from $G_n $ by more than $\Omega(m^{\frac{K+1}{2}})$.

\begin{claim}\label{claim:concentration}
$\Pr(\|\wmg(P') - G_n \|_\infty > \Omega(m^{\frac{K+1}{2}}))< \exp^{-\Omega(m)}$.
\end{claim}
\begin{proof} We first show that for each pairwise comparison $b$ vs.~$c$, $\Pr(|w_{P'}(b,c) - w_G(b,c) \cdot \frac{m^{K}}{|P_G^{\Theta}|}| = \Omega(m^{\frac{K+1}{2}}))< \exp^{-\Omega(m)}$, then apply union bound to all pairwise comparisons. Notice that $w_{P'}(b,c)$ can be viewed as the sum of $n$ independent (not necessarily identical) bounded random variables, each of which corresponds to the pairwise comparison between $b$ and $c$ in a ranking---if $b\succ c$ in the ranking then the random variable takes $1$, otherwise the random variable takes $-1$. By Hoeffding's inequality for bounded random variables, we have:
\begin{align*}
&\Pr(|w_{P'}(b,c) - \expect(w_{P'}(b,c)) |> \Omega(m^{\frac{K+1}{2}}))\\
&<\exp\{-\frac{\Omega(m^{\frac{K+1}{2}}/n)^2 n^2}{4n}\} = \exp\{-\Omega(m)\}
\end{align*}
Also notice that $\expect(w_{P'}(b,c)) = w_{P_G^{\Theta*}}(b,c)$ and $\|\wmg(P_G^{\Theta*}) - G_n  \|  = O(m^{5})=O(m^{\frac{K+1}{2}})$.  
\end{proof}

Suppose $(G,t)$ is a YES instance. That is, there exists a ranking $R'$ whose \kt{} distance to $G$ is no more than $t$. Due to Markov's inequality, $\alg$ returns a Kemeny ranking $R^*$ with probability $\ge \frac 23$. We now prove that $R^*$ is a solution to $(G,t)$ with probability $1-\exp^{-\Omega(m)}$.

Note that  for any ranking $R$, $|\kt(R,P') - \kt(R,G_n)| =  O(m^{\frac{K+5}{2}})$ holds with probability $1-\exp^{-\Omega(m)}$, which follows after applying Claim~\ref{claim:concentration} and the union bound to all  $\Theta(m^2)$ pairwise comparisons. In this case $R^*$ is a solution to $(G,t)$, because if it is not, then  $\kt(R^* ,G_n)-\kt(R',G_n)> \frac{m^K}{|P_G^{\Theta}|} =\Omega(m^{K-k-2})$. We note that $|\kt(R',P') - \kt(R',G_n)| =  O(m^{\frac{K+1}{2}})\times O(m^2)=O(m^{\frac{K+5}{2}})$ and $|\kt(R^*,P') - \kt(R^*,G_n)| =  O(m^{\frac{K+5}{2}})$, which means that $\kt(R^* ,P')-\kt(R',P')=\Theta(m^{K-k-2}) - 2 O(m^{\frac{K+5}{2}})>0$ (because $K=11+2k$), which contradicts the optimality of $R^*$. 

Therefore, Algorithm~\ref{alg:rpalg} returns YES with probability at least $\frac23-\exp^{-\Omega(m)}$, which proves that \efas{} is in {\sc RP}. Since \efas{} is NP-hard~\cite{Perrot2015:Feedback} and RP$\subseteq$NP. It follows that RP$=$NP.
\end{proof}

We now prove a similar theorem for \slaterranking{} with an additional condition.

{\em (iv)  There exist constants $k^*\ge 0$ and $B>0$ such that for any $m\ge 3$, there exist $\pi_{co}\in \Pi_m$ such that $\wmg(\pi_{co})$ has a co-cycle component $\gco$ with $\wmg(\pi_{co})\cdot \gco>\frac{B}{m^{k^*}}$.}

Condition (iv) can be viewed as the co-cycle counterpart to (iii) in Assumption~\ref{asmpt:main}.  


\begin{thm}
[\bf Smoothed Hardness of \slaterranking{}] \label{thm:slaterhard} For any series of single-agent preference models $\vec\mm$ that satisfies Assumption~\ref{asmpt:main} and (iv), if there exists a smoothed poly-time algorithm for \slaterranking{}  w.r.t.~$\vec\mm$, then {\sc NP}$=${\sc RP}.
\end{thm}
\begin{sketch} The high-level idea of the proof is similar to the proof of Theorem~\ref{thm:kemenyhard}. The difference is that in this proof we use the {\sc Tournament Feedback Arc Set (TFAS)} problem, which is NP-hard~\cite{Alon06:Ranking,Conitzer06:Slater}. Since it is unknown whether {\sc  Eulerian Tournament Feedback Arc Set} is NP-hard, in a \tfas{} instance $(G,t)$ it is possible that $G$ is not Eulerian. Therefore, we need a co-cycle component to construct a parameter profile whose WMG is $G$. Moreover, the weight on the co-cycle component cannot be too small, otherwise the construction will not be polynomial. Condition (iv) is used to guarantee the existence of a desirable co-cycle component as described. 

Slightly more formally, the proof proceeds in four steps. Step 1 is the same as the Step 1 in the proof of Theorem~\ref{thm:kemenyhard}. In Step 2, we use permutations of $\pico$, which is guaranteed by (iv), to construct a parameter profile whose WMG is a co-cycle. The reduction from \tfas{} will be presented in Step 3. In Step 4 we show that a smoothed poly-time algorithm for \slaterranking{} can be used to prove that \tfas{} is in RP.
\end{sketch}

The following example shows that Assumption~\ref{asmpt:main} and (iv) hold for a large class of Mallows-based 
models. 
\begin{ex}
\label{ex:assumption-hold-mallows}We show that for any $\underline\varphi \ne 1$, $\vec \mm_{\mallows}^{[\underline\varphi,\overline\varphi]}$ satisfies Assumption~\ref{asmpt:main} and (iv). (i) and (ii) have been discussed in Example~\ref{ex:singlemallows}. For any $\varphi\in [\underline\varphi,\overline\varphi]\cap (0,1)$, we show that (iii) and (iv) hold for $\pi = (a_1\succ \cdots\succ a_m, \varphi)$.  

For (iii), let $\gcyc = a_1\ra a_2\ra a_3\ra a_1$. \citet{Mallows57:Non-null} proved that under Mallows' model, the probability for $a\succ b$ only depends on $m,\varphi$, and the difference in the ranks of $a$ and $b$ in the central ranking. Therefore, for any $m$, $\wmg(\pi)\cdot \gcyc = 2\cdot\frac{1}{1+\varphi} - \frac{1+2\varphi}{(1+\varphi)(1+\varphi+\varphi^2)}= \frac{1+2\varphi^2}{(1+\varphi)(1+\varphi+\varphi^2)}=\Theta(1)$, see Figure~\ref{fig:wmg}. This means that $k=0$ and $A=\Theta(1)$.

For (iv), it is not hard to verify that $k^* = 0$ and $B=\Theta(1)$ for the co-cycle centered at $a_1$.
\end{ex}

\section{Parameterized Typical-Case Smoothed Complexity of \kr{}}
For any $n$-profile $P$, {\em the average \kt{} distance} is defined as
$$\avgkt(P) = \frac{1}{n(n-1)}\sum_{R_1,R_2\in P}\kt(R_1,R_2)$$ 
Following the idea in {\em probably polynomial smoothed complexity} of perceptron~\cite{Blum2002:Smoothed}, we prove a similar result on the parameterized typical-case smoothed complexity of \citet{Betzler09:Fixed}'s dynamic programming algorithm for \kr{}, denoted by $\alg_{\text{KS}}$.  
Given $\vec\mm_{\mallows}$, for convenience, sometimes we rewrite a parameter $\vec\theta\in \Theta_m^n$ as $(P,\vec\varphi)$, where $P\in\ml(\ma_m)^n$ is called the {\em central profile} and $\vec\varphi\in (0,1]^n$ is called the {\em dispersion vector}.

\begin{thm}[\bf Parameterized Typical-Case Smoothed Complexity of \kr{}] \label{thm:kemenyeasy} 
For any $m\ge 3$, $n\ge 1$, any central profile $P\in \ml(\ma_m)^n$, any dispersion vector $\vec\varphi \in (0,1]^n$, and any $t>0$, let $\varphi^* = \frac1n\sum_{j\le n}\min( m^2\varphi_j, \frac{m\varphi_j}{(1-\varphi_j)^2(1-\varphi_j^2)})$, let $d=\lceil\avgkt(P)+2\varphi^*+t\rceil$, and let $P'$ denote a random profile generated from $\mm_{\mallows,m}$ given $(P,\vec\varphi)$. With probability at least $1-\exp(-\frac{2nt^2}{m^2(m-1)^2})$, we have $\rt{\alg_{\text{KS}}}(P')=O(16^d(d^2n^2m^2\log m))$.
\end{thm}
\begin{proof} 

\citet{Betzler09:Fixed} proved that the runtime of $\alg_{\text{KS}}$ is $O(16^{\bar d}({\bar d}^2n^2m^2\log m))$, where $\bar d = \lceil \avgkt(P')\rceil$. Therefore, it suffices to prove that with high probability, the average \kt{} distance in $P'$. 


We first prove a claim on the expected \kt{} distance between the central ranking and randomly generated ranking.

\begin{claim}
\label{claim:expktsingle}
For any single-agent Mallows model and any $\theta = (R,\varphi)$, we have
$$\expect_{W\sim\theta}\kt(R,W) \le \min( m^2\varphi, \frac{m\varphi}{(1-\varphi)^2(1-\varphi^2)})$$
\end{claim}
The proof is done by directly calculating $\expect_{W\sim\theta}\kt(R,W)$ using \citet{Mallows57:Non-null}'s closed-form formulas for probability of pairwise comparisons.


\begin{claim}\label{claim:ktprob}
We have:
$$\Pr(\avgkt(P')>\avgkt(P)+2\varphi^*+t)\le \exp\left(-\frac{2nt^2}{m^2(m-1)^2}\right)$$
\end{claim}
\begin{proof}
For any pair of agents $1\le j_1<j_2\le n$, let $R_{j_1}'$ and $R_{j_2}'$ denote their rankings in $P'$, respectively. We have $\kt(R_{j_1}',R_{j_2}')\le \kt(R_{j_1},R_{j_2})+\kt(R_{j_1}',R_{j_1})+\kt(R_{j_2}',R_{j_2})$. Therefore, $\expect(\avgkt(P')) = \avgkt(P)+ \frac{(n-1)\sum_{j=1}^n \kt(R_{j}',R_{j})}{n(n-1)/2}= \avgkt(P)+ \frac{2\sum_{j=1}^n \kt(R_{j}',R_{j})}{n}$. Notice that for each $j\le n$, $\kt(R_{j}',R_{j})$ is a random variable  in  $[0,m(m-1)/2]$ whose mean is no more than $\min( m^2\varphi_j, \frac{m\varphi_j}{(1-\varphi_j)^2(1-\varphi_j^2)})$. Let $S_n =  2\sum_{j=1}^n \kt(R_{j}',R_{j})$. We have that $\expect(S_n)\le 2\varphi^*n$. Therefore, for any $t>0$, by Hoeffding's inequality, we have:
\begin{align*}
&\Pr(\frac{S_n}{n}>2\varphi^* + t)\le \Pr(\frac{S_n}{n}>\expect(\frac{S_n}{n}) + t)\\
\le  &\Pr(S_n>\expect(S_n) + nt)\le \exp\left(-\frac{2(nt)^2}{n(m(m-1))^2}\right)\\
=&\exp\left(-\frac{2nt^2}{m^2(m-1)^2}\right)
\end{align*}
\end{proof}
The theorem follows after Claim~\ref{claim:ktprob}.
\end{proof}
We note that $\varphi^*\le \frac{m^2}{n} \sum_{j=1}^n\varphi_j$ and $\varphi^*\le \frac m n\sum_{j=1}^n\frac{\varphi}{(1-\varphi)^2(1-\varphi^2)}$, and the former upper bound is the average dispersion of the agents multiplied by $m^2$. Neither upper bound implies the other. The former is stronger when some $\varphi_j$ is close to $1$. The latter is stronger when all $\varphi_j=O(\frac{1}{m})$. We immediately obtain the following corollary by combining the former upper bound with Theorem~\ref{thm:kemenyeasy}.

\begin{coro}\label{coro:ks}
Given $\vec\mm_{\mallows,m}$, any $n=\Omega(m^2(m-1)^2)$, and any parameter $(P,\vec\varphi)$ such that (1) $\avgkt(P)=O(\log m +\log n)$, and (2) $\frac{\vec\varphi\cdot\vec 1}{n}=O(\frac{\log m +\log n}{m^2})$, we have:
$$\Pr_{P'\sim (P,\vec\varphi)}(\rt{\alg_\text{KR}}(P') = \omega(\text{poly}(mn)^t)) =\exp(-\Omega(t^2))$$
\end{coro}
Corollary~\ref{coro:ks} states that when $n$ is sufficiently large, and the average \kt{} distance in $P$   and the average dispersion are not too large, with high probability $\alg_\text{KR}$ solves \kr{} in polynomial time.

\section{Summary and Future Work} 
We show that the smoothed complexity framework by~\citet{Blaser2015:Smoothed} may not be appropriate for computational social choice. We prove the smoothed hardness of Kemeny and Slater, and a parameterized typical-case smoothed easiness result for Kemeny. An immediate open question is the smoothed complexity of the Dodgson rule and the Young rule. Smoothed complexity of other problems in computational social choice is also an obvious open question as~\citet{Baumeister2020:Towards} pointed out.
\newpage
\section*{Broader Impact}
This paper aims to understand the smoothed complexity of computing commonly-studied voting rules, which are important tools for collective decision making. The results will be important to multi-agent systems, where voting is used to achieve consensus. Success of the research will benefit general public beyond the CS research community because voting is a key component of democracy.

\bibliography{references}
\newpage

\section{Appendix: Plackett-Luce-Based Models}
\label{app:pl}

\begin{ex}
\label{ex:singlepl}
As another example, in the {\em single-agent Plackett-Luce model} $\mm_{\pl,m}$, $\Theta_m= \{\vec \theta\in (0,1]^m: \vec \theta\cdot \vec 1 = 1\}$. For any $\vec \theta\in\Theta_m$ and any $R=\sigma(a_1)\succ \sigma(a_2)\succ\cdots\succ\sigma(a_m)$, we have $\pi_{\vec\theta} (R) = \prod_{i=1}^{m-1}\frac{\theta_{\sigma(a_i)}}{\sum_{l=i}^m \theta_{\sigma(a_l)}}$. For any $0< \underline\theta\le \overline\theta\le 1$, we let $\mm_{\pl,m}^{[\underline\theta,\overline\theta]}$ denote the Plackett-Luce model where $\Theta_m= \{\vec \theta\in [\underline\theta,\overline\theta]^m: \vec \theta\cdot \vec 1 = 1\}$. 

It is not hard to verify that for any $0< \underline\theta < \overline\theta\le 1$, $\mm_{\pl,m}^{[\underline\theta,\overline\theta]}$ is neutral and P-samplable by using its random utility interpretation~\cite{Xia2019:Learning}.
\end{ex}

Let $\vec \mm_{\pl}^{[\underline\theta,\overline\theta]} = \{\mm_{\pl,m}^{[\underline\theta,\overline\theta]}: m\ge 3\}$.  The following example shows that $\vec \mm_{\pl}^{[\underline\theta,\overline\theta]}$ satisfies Assumption~\ref{asmpt:main} and (iv) as $\vec \mm_{\mallows}^{[\underline\theta,\overline\theta]}$ does.

\begin{ex}
\label{ex:assumption-hold-pl}We show that for any $\underline\theta<\overline\theta$, $\vec \mm_{\pl}^{[\underline\theta,\overline\theta]}$ satisfies Assumption~\ref{asmpt:main} and (iv). (i) and (ii) have been discussed in Example~\ref{ex:singlepl}. We show that (iii) and (iv) hold for $\vec\theta= (\theta_1,\ldots,\theta_m)$, where  $\theta_1 = \overline\theta$, $\theta_2 =\frac{\underline\theta+ \overline\theta}{2}$,  and for all $3\le i\le m$, $\theta_i = \underline\theta$.    

For (iii), let $\gcyc = a_1\ra a_2\ra a_3\ra a_1$. The weights on the three edges and its orthogonal decomposition are shown in Figure~\ref{fig:expl}. (iii) follows after noticing that the $a_1\ra a_2\ra a_3$ has non-zero weight.

\begin{figure}[htp]
\centering
\includegraphics[width=.5\textwidth]{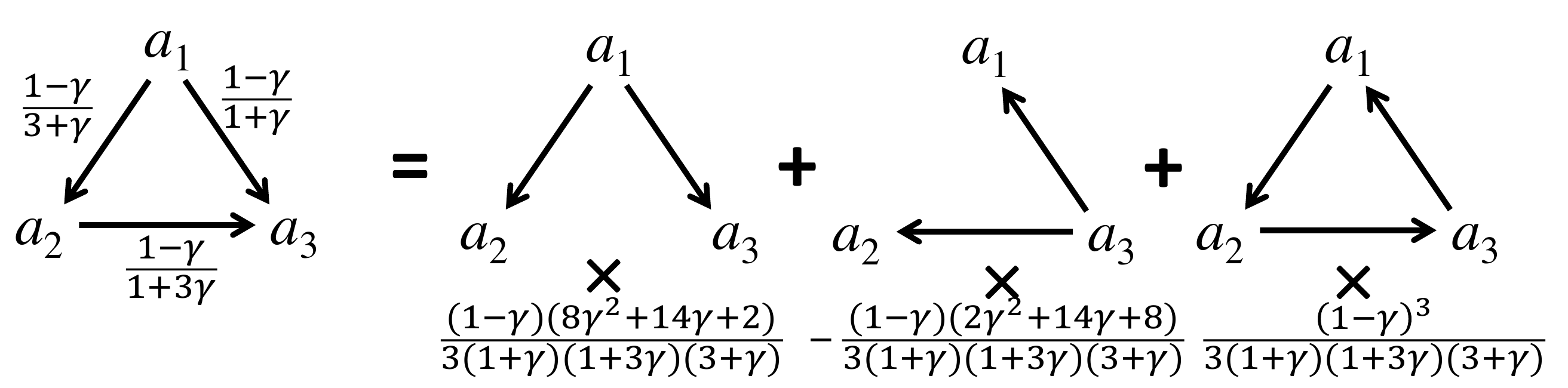} 
\caption{$\wmg(\vec\theta)$ restricted on $\gcyc$, where $\gamma = \frac{\overline \theta}{\underline\theta}$.
\label{fig:expl}}
\end{figure}

For (iv), it is not hard to verify that $k^* = 0$ and $B=\Theta(1)$ for the co-cycle centered at $a_1$ because for any $b\ne a_1$, the weight on $a_1\ra b$ is strictly positive.
\end{ex}

\section{Appendix: Proofs}
\label{app:proofs}

\appClaim{claim:qtheta}{
$|Q^\Theta_\gcyc| = O(m)$, $Q^\Theta_\gcyc$ consists of $O(m)$ types of parameters. $\wmg(Q^\Theta_\gcyc)$ consists of the following two types of edges.
(1) There are three edges $a_1\ra a_2$, $a_2\ra a_3$, $a_3\ra a_1$, each has weight $2(m-3)\alpha$, where $\alpha = \wmg(\picyc)\cdot \gcyc$.
(2) There are edges from $\{a_1,a_2,a_3\}$ to $\{a_4,\ldots,a_m\}$ whose weights are $\beta = 2\sum_{(d_1,d_2)\in \{a_1,a_2,a_3\}\times \{a_4,\ldots,a_m\}} w_{\pi_{3c}}(d_1,d_2)$.
}
\begin{proof}
    The claim follows after the definition. 
    The combination of three permutations  $\{\sigma_1^i :1\le i\le 3\}$ on $\picyc$ contributes $\alpha = \wmg(\picyc)\cdot \gcyc = w_{\picyc}(a_1,a_2)+w_{\picyc}(a_2,a_3) +w_{\picyc}(a_3,a_1)$ to each of $a_1\ra a_2$, $a_2\ra a_3$ and $a_3\ra a_1$'s weights. 
    The claim of (1) follows by considering all $6(m-3)$ permutations.
    Similarly, the combination of $2(m-3)$ permutations $\{\sigma_1^i \circ \sigma_2^t,\sigma_1^i \circ \sigma_2^{-t} :1\le t \le m-3\}$ on $\picyc$ contributes $2\sum_{d_2\in\{a_4,\cdots,a_m\}}w_{\picyc}(a_i,d_2)$ to each edge from $\{a_1,a_2,a_3\}$ to $\{a_4,\cdots,a_m\}$ for every $1 \le i \le 3$.  
    The claim of (2) follows by summation over $1 \le i \le 3$.
    Note that edges between alternatives in $\{a_4,\cdots,a_m\}$ have zero weight, since $\sigma_2^{-t}$ cancels the weight contributed by $\sigma_2^t$ for every $1 \le t \le m-3$.
\end{proof}

\appClaim{claim:patheta} {$|P_a^\Theta|=2(m-1)$ and 
$P_a^\Theta$ consists of $2(m-1)$ types of parameters. $\wmg(P^\Theta_a)$ is a co-cycle whose center is $a$ and the absolute weight on any non-zero edge is $2\sum_{(d_1,d_2)\in \{a\}\times {\ma \backslash \{a\}}}w_{\pi_{3c}}(d_1,d_2)$.
}
\begin{proof}
    It is easy to verify that the claim holds for each edge $a\ra d(d_\in \ma\backslash\{a\})$. 
    Denote the $m-1$ alternatives in $\ma \backslash \{a\}$ as  $\{a_{i_1},a_{i_2},\cdots,a_{i_{m-1}} \}$  such that $i_1\le i_2 \le \cdots \le i_{m-1}$.
    Let $r(t,i) := (t+i-1\mod (m-1)) +1, \forall 1\le i,t \le m-1$. 
    Note that $\sigma_a = a \ra a_{i_1} \ra a_{i_2} \ra \cdots \ra a_{i_{m-1}}$ and $\sigma_a \circ \eta = a \ra a_{i_{m-1}} \ra a_{i_{m-2}} \ra \cdots a_{i_1}$. Therefore, for each edge $a_{i_s}\ra a_{i_t}$ that does not contain $a$, the weight contributed by the combination of $m-1$ permutations $\{\sigma_a^{i}: 1\le i \le m-1\}$ is $$\sum_{i=1}^{m-1}w_{\wmg(\picyc)}(a_{i_{r(s,i)}},a_{i_{r(t,i)}}),$$
    while the weight contributed by the combination of $m-1$ permutations
    $\{\sigma_a^{i}\circ \eta: 1\le i \le m-1\}$ is
    \begin{align*}
        &\sum_{i=1}^{m-1}w_{\wmg(\picyc)}(a_{i_{r(m-s,i)}},a_{i_{r(m-t,i)}}) \\
         =&\sum_{i=1}^{m-1}w_{\wmg(\picyc)}(a_{i_{r(m-s,i)}},a_{i_{r(m-s+(s-t),i)}})\\
        =&\sum_{i=1}^{m-1}w_{\wmg(\picyc)}(a_{i_{r(t,i)}},a_{i_{r(t+(s-t),i)}})\\
         =&\sum_{i=1}^{m-1}w_{\wmg(\picyc)}(a_{i_{r(t,i)}},a_{i_{r(s,i)}})\\
         =&-\sum_{i=1}^{m-1}w_{\wmg(\picyc)}(a_{i_{r(s,i)}},a_{i_{r(t,i)}})
    \end{align*}
    Therefore, the weight contributed by the two combinations of permutations cancel each other, and we have $w_{\wmg(P_a^{\Theta})}(a_{i_s},a_{i_t}) = 0$.
\end{proof}

\appClaim{claim:cocycle}{$|P_1^\Theta|=O(m^2)$ and $P_1^\Theta$ consists of $O(m^2)$ types of parameters. $\wmg(P^\Theta_1)$ is a co-cycle whose center is $a_1$ and the absolute weight on any non-zero edge is $2(m-3)\beta$.
}
\begin{proof}
By claim \ref{claim:qtheta}, $|Q_{\gcyc}^{\Theta}| = O(m)$ and $Q_{\gcyc}^{\Theta}$ consists of $O(m)$ types of parameters.
Following definition \ref{dfn:op} of $O_a$, it is easy to see that $|P_1^\Theta|=O(m^2)$ and $P_1^\Theta$ consists of $O(m^2)$ types of parameters. 
By claim \ref{claim:patheta}, $\wmg(P_1^{\Theta})$ is a co-cycle and the weight on any non-zero edge $a_1\ra d(d \in \ma \backslash\{a\})$ is 
\begin{align*}
    &2\sum_{(d_1,d_2) \in \{a_1\}\times \ma \backslash \{a_1\} } w_{\wmg(\picyc)}(d_1,d_2) \\
    =& 2(w_{\wmg(\picyc)}(a_1,a_2) + w_{\wmg(\picyc)}(a_1,a_3)) + 2(m-3)\beta\\
    =& 2(m-3)\beta.
\end{align*}
The second equality follows from $w_{\wmg(\picyc)}(a_1,a_2) = -w_{\wmg(\picyc)}(a_1,a_3)$.
\end{proof}
    
\appClaim{claim:3cycle}{
$|P^\Theta_\gcyc|  = O(m^{k})$ and  
$P^\Theta_\gcyc$ consists of $O(m^3)$ different types of parameters. $\wmg(P^\Theta_\gcyc)=\gcyc$. 
}
\begin{proof}
By Claim~\ref{claim:qtheta}, $|Q^\Theta_\gcyc| = O(m)$. By Claim~\ref{claim:cocycle}, $|P^\Theta_1|=O(m^2)$. Therefore, $|P^\Theta_\gcyc| = \frac{1}{2(m-3)\alpha} O(m)+ \frac{1}{4(m-3)^2\alpha} O(m^2)= O(m^{k})$. $Q^\Theta_\gcyc$ consists of no more than $6(m-3)$ types of parameters. Each $\sigma_{a_1\leftrightarrow d}(P^\Theta_1)$ consists of $O(m^2)$ types of parameters. Therefore, the total number of parameters is no more than $O(m^3)$. 
It is not hard to verify that $\wmg(P^\Theta_\gcyc)=\gcyc$.
\end{proof}

\appThm{thm:slaterhard}{Smoothed Hardness of \slaterranking{}}{ For any series of single-agent preference models $\vec\mm$ that satisfies Assumption~\ref{asmpt:main} and (iv), if there exists a smoothed poly-time algorithm for \slaterranking{}  w.r.t.~$\vec\mm$, then {\sc NP}$=${\sc RP}.}

\begin{proof}
The high-level idea of the proof is similar to the proof of Theorem~\ref{thm:kemenyhard}. The difference is that in this proof we use the {\sc Tournament Feedback Arc Set (TFAS)} problem, which is NP-hard~\cite{Alon06:Ranking,Conitzer06:Slater}. Since it is unknown whether {\sc  Eulerian Tournament Feedback Arc Set} is NP-hard, in a \tfas{} instance $(G,t)$ it is possible that $G$ is not Eulerian. Therefore, we need a co-cycle component to construct a parameter profile whose WMG is $G$. Moreover, the weight on the co-cycle component cannot be too small, otherwise the construction will not be polynomial. Condition (iv) is used to guarantee the existence of a desirable co-cycle component as described. We conjecture that {\sc  Eulerian Tournament Feedback Arc Set} is NP-hard, and if this is true, then condition (iv) can be removed.

Formally, the theorem is proved by contradiction. Suppose for the sake of contradiction that there exists a smoothed polynomial-time algorithm $\alg_{sr}$ for \slaterranking{}. 
An instance of TFAS is denoted by $(G,t)$, where $t \in \mathbb{N}$ and $G$ is a directed unweighted tournament graph over $\ma_m$, which means that for each pair of nodes $a,b$, one of $a\ra b$  and $b\ra a$ is in $G$. We are asked to decide whether $G$ can be made acyclic by removing no more than $t$ edges.

The high-level idea of the proof is, in essence, the same as the proof of theorem \ref{thm:kemenyhard}. Formally, the proof proceeds in the following four steps. In step 1, we use permutations of $\picyc$, which is guaranteed by Assumption \ref{asmpt:main} to construct a parameter profile whose WMG is a 3-cycle whose edges have the same weights. In step 2, we use permutations of $\pico$, which is guaranteed by (iv), to construct a parameter profile whose WMG is a co-cycle. The reduction from TFAS will be presented in Step 3, where we use components of Step 1 and 2 to construct the graph. In step 4 we show that $\alg_{sr}$ can be used to prove that TFAS is in RP.

\noindent{\bf Step 1.~Construct a cycle of length 3.}
The construction is exactly the same as Step 1 in the proof of theorem \ref{thm:kemenyhard}. Let $P_{\gcyc}^{\Theta}$ be the fractional parameter profile constructed using $\theta_{3c}$. By claim \ref{claim:3cycle}, $P_{\gcyc}^{\Theta}$ consists of no more than $O(m^3)$ different types of parameters and $\wmg(P_{\gcyc}^{\Theta})$ is a 3-cycle where the weight on any non-zero edge is 1.

\noindent{\bf Step 2.~Construct a co-cycle.} 
Recall Definition \ref{dfn:op} for  $\op_a$, let $P_{a}^{\Theta} = \op_a(\theta_{co})$, where $\theta_{co}$ is the parameter corresponding to $\pico$. Similar to Claim~\ref{claim:cocycle}, the following claim characterizes $|P_a^{\Theta}|$ and $\wmg(P_a^{\Theta})$.
\begin{claim}\label{claim:slater-cocycle}
$P_a^\Theta$ consists of $2(m-1)$ parameters, each has weight $1$. $\wmg(P^\Theta_a)$ is a co-cycle whose center is $a$ and the weight on any non-zero edge is $2\gamma$, where $\gamma = \wmg(\pico)\cdot \gco$.
\end{claim}

\noindent{\bf Step 3.~Use $P_a^{\Theta}$ and $P_{\gcyc}^{\Theta}$ as building blocks to efficiently construct any WMG.}
We first show that any edge can be constructed by polynomially many (in $m$) unweighted cycles and unweighted co-cycles.

\begin{claim}\label{claim:oneedge}Any edge with weight $m$ can be obtained from the union of $m-2$ unweighted cycles and $m$ unweighted co-cycles.
\end{claim}
\begin{proof} Following the definition of the orthogonal decomposition and the fact that the $\mV_\text{cyc}$ and $3$-cycles and co-cycles discussed, any edge with weight $m$ can be obtained from unions of cycles and co-cycles. However it is unclear whether polynomially many cycles and co-cycles suffices, which is crucial later in the proof. 

The proof is done by explicitly constructing the cycles and co-cycles. It is not hard to verify that the edge $a_2\ra a_1$ with weight $m$ is the union of the following $3$-cycles and  co-cycles.
\begin{itemize}
\item {\bf Cycles.} For each $3\le i\le m$, there is a $3$-cycle $a_1\ra a_i\ra a_2\ra a_1$.
\item {\bf Co-cycles.} For each $3\le i\le m$, there is a co-cycle centered at $a_i$. There are two co-cycles centered at $a_2$.
\end{itemize}
An example for $m=4$ is shown in Figure~\ref{fig:construction}.
\begin{figure}[htp]
\centering
\includegraphics[width=.5\textwidth]{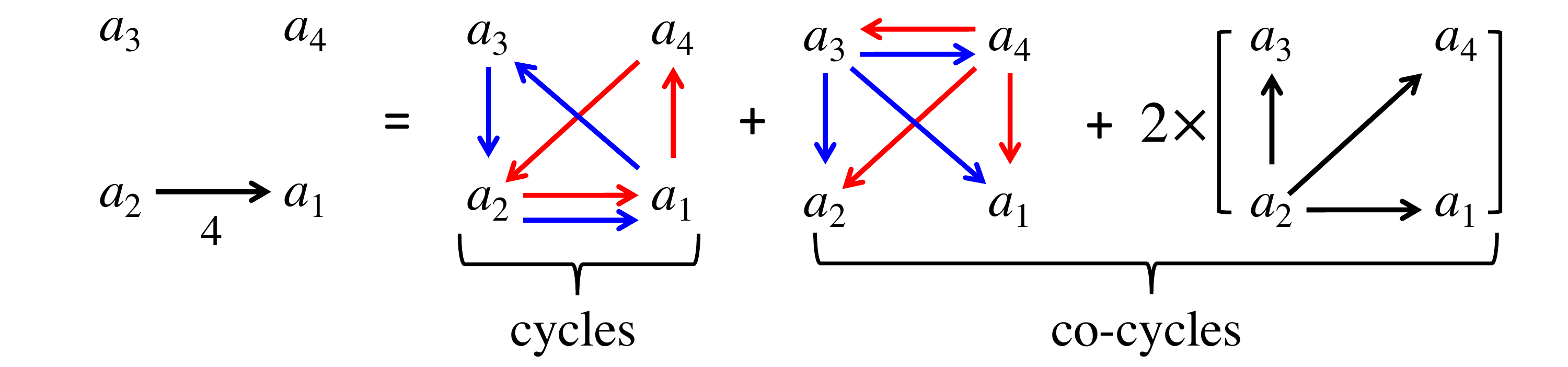}
\caption{Construction of $4$ copies of $a_2\rightarrow a_1$.\label{fig:construction}}
\end{figure}
\end{proof}

For any edge $b\ra c$ with weight $1$, let $P_{b\ra c}^\Theta$ denote the (fractional) parameter profile corresponding to the cycles and co-cycles used in the proof of Claim~\ref{claim:oneedge}. We have the following claim about $P_{b\ra c}^\Theta$.
\begin{claim}\label{claim:slater-edge}
Let $k_{\max}=\max\{k,k^*\}$. For any edge $b \ra c$, $|P^\Theta_{b\ra c}| = O(m^{k_{\max}+2})$. $|P^\Theta_{b\ra c}|$ consists of no more than $O(m^4)$ different types of parameters. $\wmg(P^\Theta_{b\ra c})$ contains a single edge $b \ra c$ with weight $m$.
\end{claim}
\begin{proof}
By Claim \ref{claim:slater-cocycle}, the $m$ unweighted co-cycles contribute no more than $m\frac{1}{2\gamma}(2(m-1))$ to $|P^\Theta_{b\ra c}|$. 
By Claim \ref{claim:3cycle}, the $(m-2)$ 3-cycles contribute no more than $(m-2)\cdot O(m^k) = O(m^{k+1})$ to $|P^\Theta_{b\ra c}|$. Therefore, $|P^\Theta_{b\ra c}| \le m\frac{1}{2\gamma}(2(m-1))+ O(m^{k+1})= O(m^{k^*+2})+O(m^{k+1}) = O(m^{k_{\max}+2})$. 

Similarly, the number of different types of parameters in $|P^\Theta_{b\ra c}|$ is no more than $m\times2(m-1) + (m-2)\times O(m^3) = O(m^4)$.
\end{proof}

By Claim \ref{claim:slater-edge}, each unweighted majority graph can be represented by the $\wmg$ of a fractional parameter profile of no more than $O(m^6)\times {m\choose 2} = O(m^6)$ types of parameters, whose total weight is no more than $O(m^{k_{\max}+2})$, formally defined as follows. 
\begin{dfn}
For any $\umg$ $G$, let $P^\Theta_{G} =\bigcup_{a\ra b \in \umg} P_{b\ra c}^\Theta$. 
\end{dfn}
\noindent{\bf Step 4.~Use the smoothed polynomial-time algorithm to solve {\sc TFAS}.} 
Let $K = 2k_{\max} + 7 > 6$. Given an {\sc TFAS} instance $(G,t)$, where $G$ is an unweighted directed tournament graph, we first define an integral profile of parameters $P^{\Theta*}_{G}$ of $n = \Theta(m^K)$ parameters that is approximately $\frac{m^K}{|P^\Theta_{G}|}$ duplicates of $P^\Theta_G$ with error no more than $O(m^6)$. More precisely, let
$$
P_G^{\Theta*} = \lfloor P_G^\Theta \cdot \dfrac{m^{K}}{|P_G^{\Theta}|}\rfloor.
$$
Let $n = |P^{\Theta*}_G|$. Because the number of different types of parameters in $P_G^{\Theta*}$ is $O(m^6)$, we have $n = m^{K}-O(m^6)$, $\|\wmg(P_G^{\Theta*}) - \wmg(P_G^{\Theta} \cdot \frac{m^{K}}{|P_G^{\Theta}|}) \|  = O(m^6)$, and $\|\wmg(P_G^{\Theta*}) - G \cdot \frac{m^{K}}{|P_G^{\Theta}|}) \|  = O(m^6)$.

We now prove that running the smoothed polynomial-time algorithm $\alg_{sr}$ on instance generated by $P_G^{\Theta*}$ will give us a correct answer to {\sc TFAS} with high probability. The idea is similar to Step 3 in the proof of Theorem~\ref{thm:kemenyhard}. We will  simulate the preference profile generated under $P_G^{\Theta*}$: for each parameter $\theta'$ in $P_G^{\Theta*}$, we independently sample a linear order $R_{\theta'}$ from $\mm$ given $\theta'$. Let $P'$ denote the resulted preference profile, which is a random variable. The next claim shows that with exponentially small probability $\umg(P')$ is different from $G$.
\begin{claim}\label{claim:slater-concentration}
$\Pr[\umg(P') \ne G] < \exp\{-\Omega(m)\}$.
\end{claim}
\begin{proof}
We first show that for each pairwise comparison $b$ vs.~$c$, $\Pr[|w_{P'}(b,c)-\E[w_{P'}(b,c)]| >E[w_{P'}(b,c)]] < \exp\{-\Omega(m)\}$, then apply union bound to all pairwise comparisons. Notice that $w_{P'}(b,c)$ can be viewed as the sum of $n$ independent (but may not necessarily identical) bounded random variables, each of which corresponds to the pairwise comparison between $b$ and $c$---if $b\succ c$ in the linear order then the random variable takes $1$, otherwise the random variable takes $-1$. By Hoeffding's inequality for bounded random variables, we have: 
\begin{align*}
&\Pr[|w_{P'}(b,c)-\E[w_{P'}(b,c)]| >E[w_{P'}(b,c)]] \\
< &2\exp\{-\frac{(\E[w_{P'}(b,c)])^2}{2n}\}\\
= &2\exp\{-\frac{m^{2(K-k_{\max}-2+1)}}{2m^k}\}
=\exp\{-\Omega(m)\}
\end{align*}
Also notice that $\E[w_{P'}(b,c)] = m\cdot\frac{m^{K}}{|P^{\Theta|}_G}\cdot w_{G}(b,c)$ where $w_{G}(b,c)\in \{1,-1\}$.  Thus with high probability, $w_{\umg(P')}(b,c) = w_{G}(b,c)$. The claim follows after applying the union bound to all $\binom{m}{2} = \Theta(m^2)$ pairwise comparisons.
\end{proof}
To see that  {\sc TFAS} is in {\sc RP}, we can run Algorithm~\ref{alg:rpalg} with $\alg = \text{Alg}_{\text{sr}}$. More precisely, we run algorithm $\text{Alg}_{\text{sr}}$ for three times of its expected runtime, and if it is not finished yet, return NO. Otherwise, we check whether $\text{Alg}_{\text{sr}}(P')$ gives a solution to $(G,t)$. The rest of the proof is similar to the proof of Theorem~\ref{thm:kemenyhard}.
\end{proof}

\appClaim{claim:expktsingle}{
For any single-agent Mallows model and any $\theta = (R,\varphi)$, we have
$$\expect_{W\sim\theta}\kt(R,W) \le \min( m^2\varphi, \frac{m\varphi}{(1-\varphi)^2(1-\varphi^2)})$$
}
\begin{proof}
According to the calculation by~\citet{Mallows57:Non-null}, for any pair of alternatives $(a,b)$ ranked at the $i_1$-th$<i_2$-th positions in the central ranking, respectively, $\Pr_{W\sim\theta}(a\succ_W b) = \frac{i_2-i_1+1}{1-\varphi^{i_2-i_1+1}} - \frac{i_2-i_1}{1-\varphi^{i_2-i_1}}$, which is at least $\frac{1}{1+\varphi}$. Therefore, $\expect_{W\sim\theta}\kt(R,W) \le \frac{m^2}{1+\varphi} = \sum_{a\succ_R b}\Pr_{W\sim\theta}(b\succ_W a) \le \frac{m^2\varphi}{1+\varphi}\le m^2\varphi$.

The $\frac{m\varphi}{(1-\varphi)^2(1-\varphi^2)}$ part needs a more careful calculation. For each pair of alternatives whose positions in $R$ are $k$ away, we have 
\begin{align*}
&\Pr_{W\sim\theta}(b\succ_W a) = 1-\frac{k+1}{1-\varphi^{k+1}}+\frac{k}{1-\varphi^{k}}\\
=& \frac{k\varphi^k-(k+1)\varphi^{k+1}+\varphi^{2k+1}}{(1-\varphi^k)(1-\varphi^{k+1})}\\
<& \frac{k\varphi^k-(k+1)\varphi^{k+1}+\varphi^{2k+1}}{(1-\varphi)^2}
\end{align*}
Summing up the probability for all pairs, we have
\begin{align*}
&\expect_{W\sim\theta}\kt(R,W)\\
<& (m-1)\sum_{k=1}^{m-1}\frac{k\varphi^k-(k+1)\varphi^{k+1}+\varphi^{2k+1}}{(1-\varphi)^2}
\\
=& \frac{m-1}{(1-\varphi)^2}\left(\sum_{k=1}^{m}\varphi^{2k-1}-m\varphi^m\right)\\
=& \frac{(m-1)\varphi}{(1-\varphi)^2}\left(\sum_{k=1}^{\lfloor m/2\rfloor} \varphi^{2(k-1)}(1-\varphi^{m+1-2k})^2 \right)\\
\le& \frac{(m-1)\varphi}{(1-\varphi)^2(1-\varphi^2)} <\frac{m\varphi}{(1-\varphi)^2(1-\varphi^2)} 
\end{align*}
\end{proof}

\end{document}